\documentclass[%
 reprint,
superscriptaddress,
 amsmath,amssymb,
 aps,
 prl,
]{revtex4-1}
\usepackage[utf8]{inputenc}
\usepackage[T1]{fontenc}
\usepackage{mathpazo}
\usepackage{graphicx}
\usepackage{dcolumn}
\usepackage{bm}
\usepackage{braket}

\usepackage[colorlinks,citecolor=blue,linkcolor=blue]{hyperref}

\newtheorem{theorem}{Theorem}
\newtheorem{lemma}[theorem]{Lemma}
\newtheorem{proposition}[theorem]{Proposition}

\newenvironment{proof}[1][Proof]{\begin{trivlist}
\item[\hskip \labelsep {\bfseries #1}]}{\end{trivlist}}
\newenvironment{definition}[1][Definition]{\begin{trivlist}
\item[\hskip \labelsep {\bfseries #1}]}{\end{trivlist}}

\newcommand{\qed}{\nobreak \ifvmode \relax \else
      \ifdim\lastskip<1.5em \hskip-\lastskip
      \hskip1.5em plus0em minus0.5em \fi \nobreak
      \vrule height0.75em width0.5em depth0.25em\fi}

\usepackage{graphicx}
\graphicspath{ {../images/} }

\usepackage{bbm}
\usepackage{soul}
\usepackage[normalem]{ulem}

\begin{document}


\title{Exponential Communication Complexity Advantage from Quantum Superposition of the Direction of Communication}

\author{Philippe Allard Gu\'{e}rin}
\affiliation{Faculty of Physics, University of Vienna, Boltzmanngasse 5, 1090 Vienna, Austria}
\affiliation{Institute for Quantum Optics and Quantum Information (IQOQI), Austrian Academy of Sciences, Boltzmanngasse 3, 1090 Vienna, Austria}

\author{Adrien Feix}
\affiliation{Faculty of Physics, University of Vienna, Boltzmanngasse 5, 1090 Vienna, Austria}
\affiliation{Institute for Quantum Optics and Quantum Information (IQOQI), Austrian Academy of Sciences, Boltzmanngasse 3, 1090 Vienna, Austria}

\author{Mateus Ara\'{u}jo}
\affiliation{Faculty of Physics, University of Vienna, Boltzmanngasse 5, 1090 Vienna, Austria}
\affiliation{Institute for Quantum Optics and Quantum Information (IQOQI), Austrian Academy of Sciences, Boltzmanngasse 3, 1090 Vienna, Austria}

\author{\v{C}aslav Brukner}
\affiliation{Faculty of Physics, University of Vienna, Boltzmanngasse 5, 1090 Vienna, Austria}
\affiliation{Institute for Quantum Optics and Quantum Information (IQOQI), Austrian Academy of Sciences, Boltzmanngasse 3, 1090 Vienna, Austria}

\date{\today}

\begin{abstract}
In communication complexity, a number of distant parties have the task of calculating a distributed function of their inputs, while minimizing the amount of communication between them. It is known that with quantum resources, such as entanglement and quantum channels, one can obtain significant reductions in the communication complexity of some tasks. In this work, we study the role of the quantum superposition of the direction of communication as a resource for communication complexity.  We present a tripartite communication task for which such a superposition allows for an exponential saving in communication, compared to one-way quantum (or classical) communication; the advantage also holds when we allow for protocols with bounded error probability.

\end{abstract}

\maketitle

Quantum resources make it possible to solve certain communication and computation problems more efficiently than what is classically possible. In communication complexity problems, a number of parties have to calculate a distributed function of their inputs while reducing the amount of communication between them~\cite{Yao1979, Kushilevitz1996}. The minimal amount of communication is called the \textit{complexity of the problem}. For some communication complexity tasks, the use of shared entanglement and quantum communication significantly reduces the complexity as compared to protocols exploiting shared classical randomness and classical communication~\cite{Yao1993, Buhrman2010}. Important early examples for which quantum communication yields an exponential reduction in communication complexity over classical communication are the distributed Deutsch-Jozsa problem~\cite{Buhrman1998} and Raz's problem~\cite{Raz1999}.

Quantum computation and communication are typically assumed to happen on a definite causal structure, where the order of the operations carried on a quantum system is fixed in advance. However, the interplay between general relativity and quantum mechanics might force us to consider more general situations in which the metric, and hence the causal structure, is indefinite. Recently, a quantum framework has been developed with no assumption of a global causal order~\cite{Oreshkov2012, Araujo2015, Oreshkov2015}. This framework can also be used to study quantum computation beyond the circuit model, for instance using the ``quantum switch'' as a resource --- a qubit coherently controlling the order of the gates in a quantum circuit~\cite{Chiribella2013}. It has recently been realized experimentally~\cite{Procopio2015}.

It was shown that this new resource provides a reduction in complexity to $n$ black-box queries in a problem for which the optimal quantum algorithm with fixed order between the gates requires a number of queries that scales as $n^2$~\cite{Araujo2014_PRL}. The quantum switch is also useful in communication complexity; a task has been found for which the quantum switch yields an increase in the success probability, yet no advantage in the asymptotic scaling of the communication complexity was found~\cite{Feix2015}. Most generally, no information processing task is known for which the quantum switch (or any other causally indefinite resource) would provide an exponential advantage over causal quantum (or classical) algorithms. 

Here we find a tripartite communication complexity task for which there is an exponential separation in communication complexity between the protocol using the quantum switch and any causally ordered quantum communication scheme. The task requires no promise on inputs and is inspired by the problem of deciding whether a pair of unitary gates commute or anticommute, which can be solved by the quantum switch with only one query of each unitary~\cite{Chiribella2012}. If the parties are causally ordered, the number of qubits that needs to be communicated to accomplish the task scales linearly with the number of input bits, whereas the protocol based on the quantum switch only requires logarithmically many communicated qubits. This shows that causally indefinite quantum resources can provide an exponential advantage over causally ordered quantum resources (i.e., entanglement and one-way quantum channels).

The tripartite causally ordered communication scenario we consider in this paper is illustrated in Fig.~\ref{fig:causally_ordered}. Alice and Bob are respectively given inputs $x \in X$ and $y \in Y$, taken from finite sets $X$, $Y$. There is a third party, Charlie, whose goal is to calculate a binary function $f(x,y)$ of Alice's and Bob's inputs, while minimizing the amount of communication between all three parties. We shall first assume that communication is one-way only: from Alice to Bob and from Bob to Charlie. Furthermore, we grant the parties access to unrestricted local computational power and unrestricted shared entanglement. We will also consider bounded error communication, in which the protocol must succeed on all inputs with an error probability smaller than $\epsilon$.

In quantum communication, the parties communicate with each other by sending quantum systems. Conditionally on their inputs, the parties may apply general quantum operations to the received system, and then send this system out. We require that the parties' local laboratories receive a system \textit{only once} from the outside environment. We impose this requirement to exclude sequential communication, in which the parties communicate back and forth by sending quantum systems to each other at different time steps. Alice's laboratory has an input and output quantum state, consisting of $N_{A_I}$ and $N_{A_O}$ qubits, respectively; similar notation is used for Bob's and Charlie's systems. We seek to succeed at the communication task on all inputs with error probability lower than $\epsilon$, while minimizing the number of communicated qubits $N := N_{A_O} + N_{B_O}$. The optimal causally ordered strategy is for Bob to calculate $f(x,y)$ and then communicate the result to Charlie using one bit of communication; in this case $N_{A_O}$ is a good lower bound for $N$.
\begin{figure}[t]
\includegraphics[width=7cm]{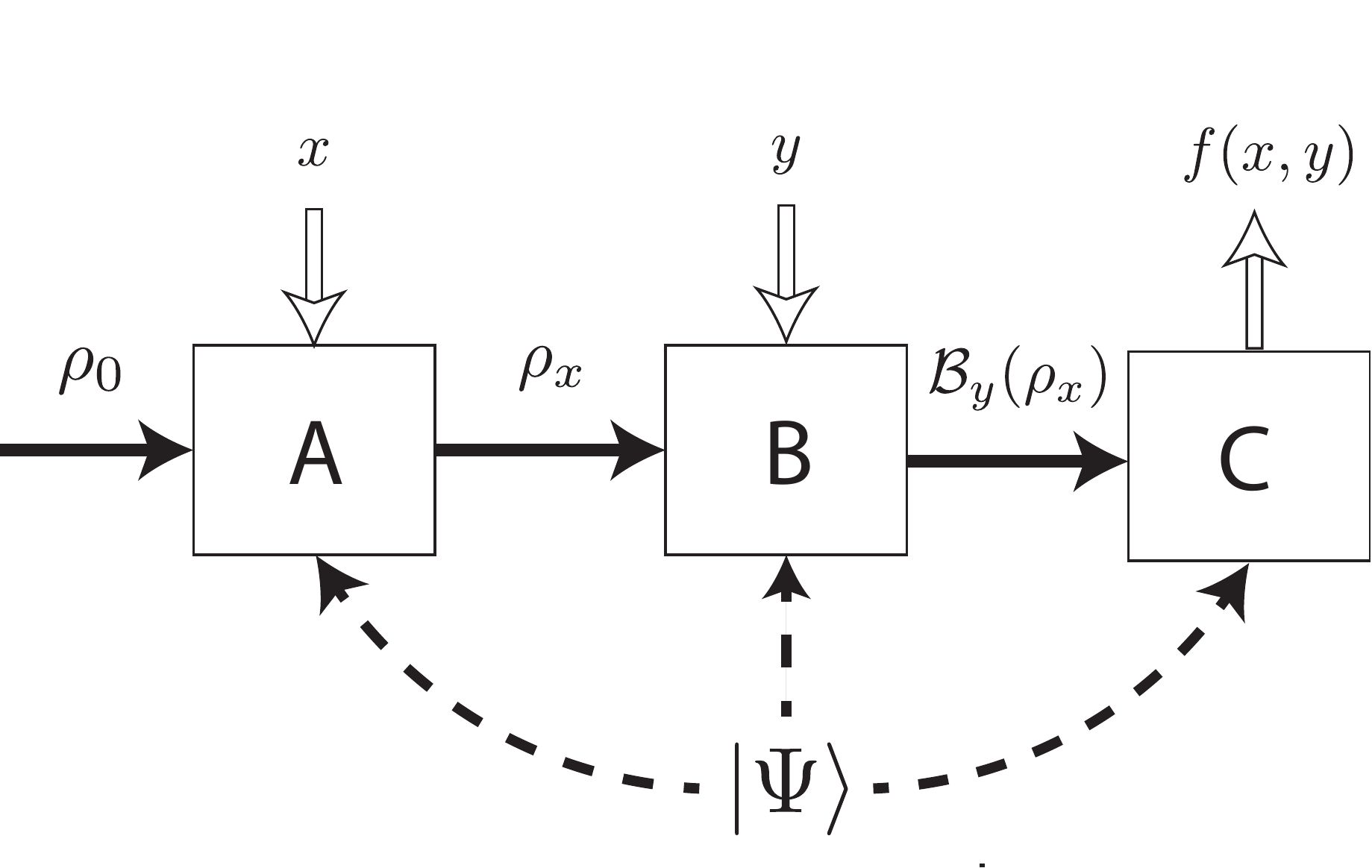}
\caption{Causally ordered quantum communication complexity scenario. Conditionally on their inputs $x$ and $y$, Alice sends a state $\rho_x$ to Bob, who then applies a CP map $\mathcal{B}_y$ and sends the system to Charlie. The unlimited entanglement shared between the parties is represented by $\ket{\Psi}$. The optimal causally ordered protocol is the one that minimizes the number of qubits in $\rho_x$ (which is a lower bound for the communication complexity of the task)}.
\centering
\label{fig:causally_ordered}
\end{figure}

The communication complexity of any causally ordered tripartite communication complexity task can be bounded by considering the bipartite task obtained by identifying Bob and Charlie as a single party. Bearing this in mind, we prove a tight lower bound on the quantum communication complexity of an important family of one-way bipartite deterministic (error probability $\epsilon = 0$) communication tasks, which in turn implies a lower bound on the communication complexity of causally ordered tripartite tasks. This result appears in Theorem 5 of Ref.~\cite{Klauck2000}, but we present a different proof here.
\begin{lemma}
\label{lemma:rows}
For deterministic one-way evaluation of any binary distributed function $f: X \times Y \to \{0,1\}$ such that $\forall x_1, x_2 \in X$, with $x_1 \neq x_2$, $\exists y \in Y$ for which $f(x_1, y) \neq f(x_2, y)$, the minimum Hilbert space dimension of the system sent between two parties sharing an arbitrary amount of entanglement is $ d = \bigl\lceil \sqrt{|X|} \, \bigr\rceil$. Equivalently, the minimum number of communicated qubits is $\lceil \log_2 d \rceil$.
\end{lemma}

\begin{proof}
We recall a well-known result of quantum information~\cite{Hausladen1996}, establishing that if Alice and Bob share unlimited entanglement, the largest number of orthogonal (perfectly distinguishable) states that Alice can transmit to Bob by sending a $d$-dimensional system is $d^2$. Therefore, they can deterministically compute $f$ if Alice sends a system of  Hilbert space dimension $\bigl\lceil \sqrt{|X|} \, \bigr\rceil$.

Suppose by way of contradiction that the Hilbert space dimension of the communicated system is only $( \bigl\lceil \sqrt{|X|} \, \bigr\rceil - 1)$. The maximal number of orthogonal states that can be transmitted by Alice to Bob is $(\bigl\lceil \sqrt{|X|} \, \bigr\rceil - 1)^2 < |X|$. Therefore, there exist inputs $x_1, x_2 \in X$ such that the corresponding states $\rho_1$, $\rho_2$ transmitted to Bob are not orthogonal, and thus not perfectly distinguishable~\cite{NielsenChuang}. By our assumption about the function $f$, there exists an input $y \in Y$ such that $f(x_1, y) \neq f(x_2, y)$. Therefore, if Bob receives the input $y$, he will need to distinguish between $\rho_1$ and $\rho_2$ in order to output the function correctly, but this cannot be done with zero error probability.
\qed
\end{proof}

The previous lemma establishes that for a very large class of deterministic communication complexity tasks, it is necessary for Alice to communicate all of her input to Bob. In these cases, the only advantage achieved by causal one-way quantum communication is a reduction by a constant factor of two due to dense coding~\cite{Bennett1992}. An important example of this form is the Inner Product game~\cite{Cleve1998, Nayak2002}. Note that Lemma~\ref{lemma:rows} does not apply to relational tasks such as the hidden matching problem \cite{Bar-Yossef2004}, for which there is an exponential separation between quantum and classical communication complexity.

\begin{figure}[t]
\includegraphics[width=8cm]{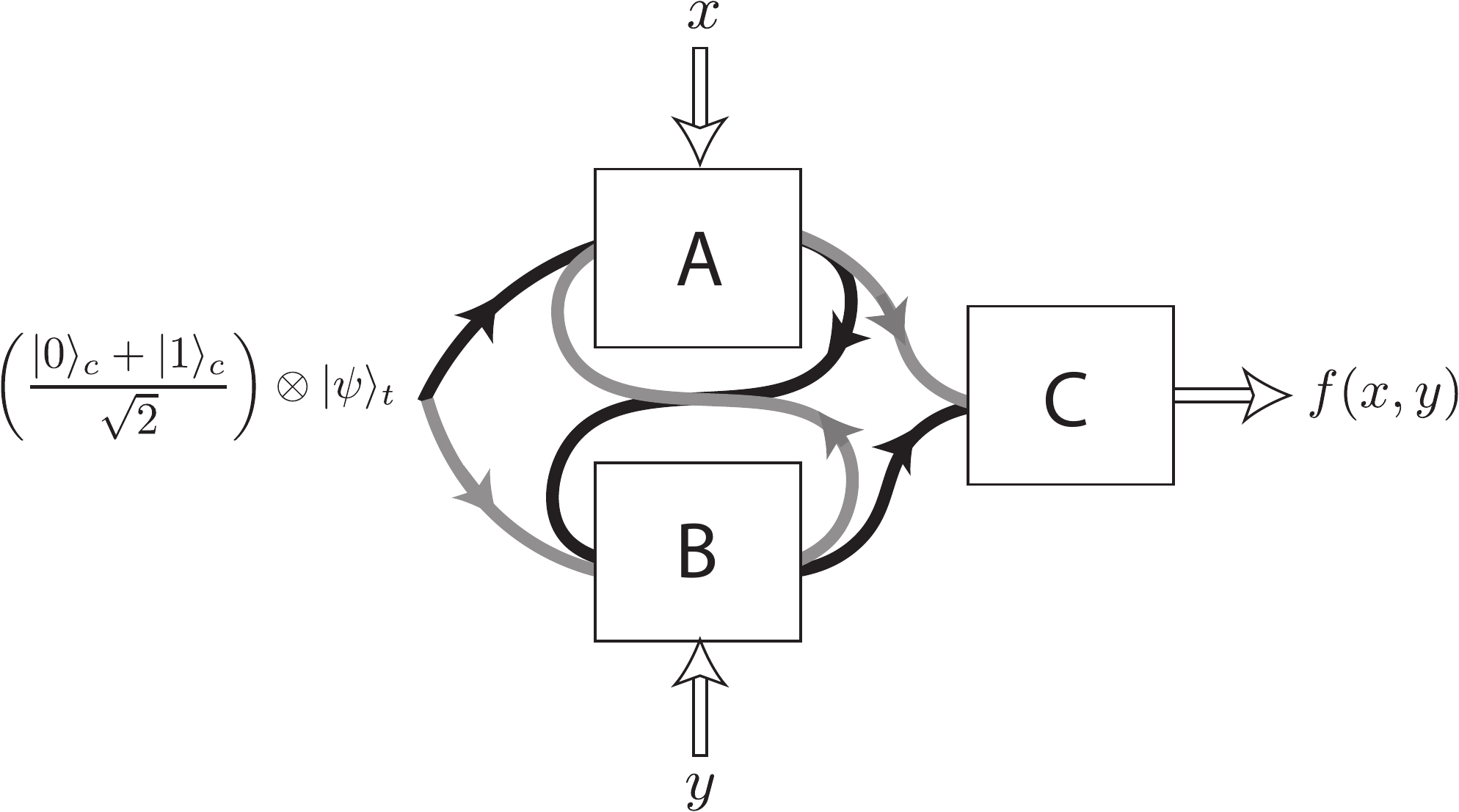}
\caption{Communication complexity setup using the quantum switch. A qubit in the state $\frac{1}{\sqrt{2}} (\ket{0}_c + \ket{1}_c) $ coherently controls the path taken by a system of $N$ qubits in initial state $\ket{\psi}_t$. One path goes first through Alice's lab and then Bob's, while the other path goes first through Bob's lab and then Alice's. Alice and Bob are given classical inputs $x\in X$, $y\in Y$, and Charlie (using the control qubit) computes a binary function of their inputs $f(x,y)$}.
\centering
\label{fig:switch}
\end{figure}

We now seek to establish a communication complexity task for which indefinite causal order can be used as a resource. In the following we assume that the parties have local laboratories, and that they receive a quantum system from the environment only once. They then perform a general quantum operation on their system, and send it out. An example of a noncausally ordered process is the quantum switch~\cite{Chiribella2013}, whose use in the context of communication complexity is shown in Fig.~\ref{fig:switch}. Charlie is in the causal future of both Alice and Bob, and an ancilla qubit coherently controls the causal ordering of Alice and Bob; both the target state and the control qubit are prepared externally. Assume that Alice and Bob apply unitary gates $U_A$ and $U_B$ to their respective input systems of $N$ qubits. The global unitary describing the evolution of the system from Charlie's point of view is 
\begin{equation}
\label{eq:V_switch}
V(U_A, U_B) = \ket{0}\bra{0}_c \otimes (U_B U_A)_t + \ket{1}\bra{1}_c \otimes (U_A U_B)_t,
\end{equation}
where the index $c$ denotes the control qubit, and the unitaries $U_A$ and $U_B$ act on the target Hilbert space of $N$ qubits.

Using the quantum switch, one can determine whether two unitaries $U_A$, $U_B$ commute or anticommute with a single query of each unitary, while at least one unitary  must be queried twice in the causally ordered case~\cite{Chiribella2012}. Explicitly, consider the quantum switch with the control qubit initially in state $\ket{+}_c = \frac{1}{\sqrt{2}} (\ket{0}_c + \ket{1}_c)$ and with initial target state $\ket{\psi}_t$. If $A$ and $B$ apply local unitaries $U_A$ and $U_B$, the resulting state after applying $V(U_A, U_B)$ is
\begin{equation}
\frac{1}{\sqrt{2}} \left( \ket{0}_c \otimes U_B U_A \ket{\psi}_t + \ket{1}_c \otimes U_A U_B \ket{\psi}_t \right).
\end{equation}
If Charlie subsequently applies a Hadamard gate to the control qubit, the resulting state is
\begin{equation}
\frac{1}{2} \left( \ket{0}_c \otimes \{ U_A, U_B \} \ket{\psi}_t - \ket{1}_c \otimes [U_A, U_B] \ket{\psi}_t \right).
\label{eq:comm_anti_state}
\end{equation}

Suppose that Alice and Bob randomly choose unitaries from a set $\mathcal{U}$ and that there exists a state $\ket{\psi}_t$ such that $\forall U,V \in \mathcal{U}$, either $[U, V]\ket{\psi}_t = 0$ or $\{U, V\} \ket{\psi}_t = 0$. Then Eq.~\eqref{eq:comm_anti_state} shows that the quantum switch with initial target state $\ket{\psi}_t$ and control qubit $\ket{+}_c$ as inputs allows Charlie to discriminate between these two possibilities with certainty by measuring the control qubit in the computational basis.

We now define a communication complexity task, the Exchange Evaluation game $EE_n$, for any integer $n$. In this game, Alice and Bob are respectively given inputs $(\mathbf{x},f), (\mathbf{y},g) \in \mathbb{Z}_2^n \times F_n$, 
where $F_n$ is the set of functions over $\mathbb{Z}_2^n$ that evaluate to zero on the zero vector
\begin{equation}
F_n = \left\{ f: \mathbb{Z}_2^n \to \mathbb{Z}_2 \, | \,  f(\mathbf{0}) = 0 \right \}.
\end{equation}
Charlie must output
\begin{equation}
EE_n(\mathbf{x}, f, \mathbf{y}, g) = f(\mathbf{y}) \oplus g(\mathbf{x}),
\end{equation}
where the symbol $\oplus$ denotes addition modulo 2. This game can be interpreted as the sum modulo 2 of two parallel random access codes~\cite{Ambainis1999}. 

We first construct an encoding of the inputs $(\mathbf{x}, f), (\mathbf{y}, g)$ in terms of local $n$-qubit unitaries that all commute or anticommute; we then use the previous observation to conclude that the switch succeeds deterministically at this task with $n$ qubits of communication. We start with some definitions. The group of Pauli $X$ operators on $n$ qubits is defined as
\begin{equation}
X(\mathbf{x}) = X_1^{x_1} \otimes X_2^{x_2} \otimes \dots \otimes X_n^{x_n}, 
\end{equation}
where $x_i$ is the $i$th component of the binary vector $\mathbf{x} \in \mathbb{Z}_2^n$. Here, $X_i$ is the single qubit Pauli $X$-operator acting on the $i$th qubit, and $X_i^0 = \mathbb{I}_i$ is the single qubit identity matrix.

We associate to every $f \in F_n$ a diagonal matrix
\begin{equation}
D(f) = \sum_{\mathbf{z} \in \mathbb{Z}_2^n} (-1)^{f(\mathbf{z})} \ket{\mathbf{z}}\bra{\mathbf{z}},
\end{equation}
where $\ket{\mathbf{z}}$ is the state such that $Z_i \ket{\mathbf{z}} = (-1)^{z_i} \ket{\mathbf{z}}$, with $Z_i$ the single qubit Pauli $Z$ operator acting on qubit $i$. The set $\{ D(f) \}_{f \in F_n}$ consists of all diagonal matrices with entries $\pm 1$ in the computational basis, such that the first entry is $+1$. 

We define the set of unitaries
\begin{equation}
\mathcal{U}_n = \{ X(\mathbf{x})D(f) | (\mathbf{x} , f) \in \mathbb{Z}_2^n \times F_n \},
\label{eq:U_n}
\end{equation}
which has dimension
\begin{equation}
|\mathcal{U}_n| = 2^{2^n +n -1}.
\label{eq:dim_U_n}
\end{equation}
This superexponential scaling of $| \mathcal{U}_n |$ is essential to establish a communication advantage with the quantum switch. Also note that
\begin{equation}
X(\mathbf{x})D(f)X(\mathbf{y})D(g) \ket{\mathbf{0}} = (-1)^{f(\mathbf{y})} \ket{ \mathbf{x \oplus y}}.
\end{equation}
Therefore, when acting on the $n$-qubit input state $\ket{\mathbf{0}}$, the elements of $\mathcal{U}_n$ all commute or anticommute with each other, and 
\begin{align*}
[X(\mathbf{x})D(f),X(\mathbf{y})D(g)] \ket{\mathbf{0}} &= 0 \,, \mathrm{if} \, \,(-1)^{f(\mathbf{y})} = (-1)^{ g(\mathbf{x})} \nonumber \\
\{X(\mathbf{x})D(f),X(\mathbf{y})D(g)\} \ket{\mathbf{0}} &= 0 \,, \mathrm{if} \, \,(-1)^{f(\mathbf{y})} = (-1)^{ g(\mathbf{x})+1} .
\end{align*}

Therefore, the game is equivalent to determining whether the corresponding unitaries $X(\mathbf{x})D(f)$ and $X(\mathbf{y})D(g)$ commute or anticommute \textit{when applied to the state} $\ket{\mathbf{0}}$. By the discussion following Eq.~\eqref{eq:comm_anti_state}, this problem can be solved deterministically by Charlie using the quantum switch with $O(n)$ qubits of communication from Alice to Bob, with a strategy consisting of applying the unitary corresponding to their input according to Eq.~\ref{eq:U_n}.

We now show that the Exchange Evaluation game satisfies the conditions of Lemma~\ref{lemma:rows}; this will allow us to conclude that for deterministic ($\epsilon = 0$) evaluation in the one-way causally ordered case, $EE_n$ requires an amount of communicated qubits that grows exponentially with $n$.
\begin{proposition}
For every $(\mathbf{x_1}, f_1), (\mathbf{x_2}, f_2) \in \mathbb{Z}_2^n \times F_n$, such that $(\mathbf{x_1},f_1) \neq (\mathbf{x_2}, f_2)$, there exists $(\mathbf{y}, g) \in \mathbb{Z}_2^n \times F_n$ such that $EE_n(\mathbf{x_1}, f_1, \mathbf{y}, g) \neq EE_n(\mathbf{x_2}, f_2, \mathbf{y},g)$.
\label{prop:EE_condition}
\end{proposition}

\begin{proof}
First note that $EE_n(\mathbf{x_1}, f_1, \mathbf{y}, g) \neq EE_n(\mathbf{x_2}, f_2, \mathbf{y}, g)$ if and only if
\begin{equation}
f_1(\mathbf{y}) \oplus f_2(\mathbf{y}) \oplus g(\mathbf{x_1}) \oplus g(\mathbf{x_2})= 1.
\label{eq:rows_condition_game}
\end{equation}
Then, since $(\mathbf{x_1}, f_1) \neq (\mathbf{x_2}, f_2)$, either $\mathbf{x_1} \neq \mathbf{x_2}$ or $f_1 \neq f_2$ holds. We check that the conditions of the lemma are satisfied in both cases.

\paragraph{(i) Case where $\mathbf{x_1} \neq \mathbf{x_2}$:\\}

Suppose without loss of generality that $\mathbf{x_1} \neq \mathbf{0}$ and define $g$ as the function such that $g(\mathbf{x_1}) = 1$ and $g(\mathbf{z}) = 0, \, \forall \mathbf{z} \neq \mathbf{x_1}$. Also, because $f_1, f_2 \in F_n$, $f_1(\mathbf{0}) = f_2(\mathbf{0}) = 0$. Therefore, the function $g$ we just defined and $\mathbf{y} = \mathbf{0}$ satisfy Eq.~\eqref{eq:rows_condition_game}.

\paragraph{(ii) Case where $f_1 \neq f_2$:\\}

Let $\mathbf{y} \in \mathbb{Z}_2^n$ be a vector for which $f_1$ and $f_2$ differ, so that $f_1(\mathbf{y}) + f_2(\mathbf{y}) = 1$. Then this $\mathbf{y}$ and the zero function $g(\mathbf{x}) = 0 \, \forall \mathbf{x}$ satisfies Eq.~\eqref{eq:rows_condition_game}.\qed
\end{proof}

According to Eq.~\eqref{eq:dim_U_n}, the dimension of the set of inputs to $EE_n$ is $|\mathcal{U}_n| = 2^{2^n +n -1}$. Direct application of Proposition \ref{prop:EE_condition} with Lemma \ref{lemma:rows} establishes that the number of qubits of communication required for deterministic success in the causally ordered case is $\frac{1}{2} \log_2 |\mathcal{U}_n|  = \frac{1}{2}(2^n + n -1) = \Omega(2^n)$, using dense coding. In comparison, we have seen that with the quantum switch as a resource, we need only $n$ qubits of communication between Alice and Bob to calculate this function. We thus conclude that for the Exchange Evaluation game, there is an exponential separation in the deterministic communication complexity of $EE_n$.

Note that with two-way (classical) communication, it is possible to solve the Exchange Evaluation game with $2n + 2$ bits of communication, simply by having Alice and Bob send their vectors $\mathbf{x}$, $\mathbf{y}$ to the other party, followed by local evaluation of $f(\mathbf{y})$ and $g(\mathbf{x})$ by the parties and communication of the result to Charlie. We emphasize that once we allow two-way communication, the quantum advantage can also disappear in traditional quantum communication complexity (comparing causally ordered quantum communication with classical communication): this is the case for the distributed Deutsch-Jozsa problem~\cite{Buhrman1998}, but not for Raz's problem~\cite{Klartag2011}.

For causally ordered communication complexity tasks, the exponential quantum-classical separation does not always continue to hold when allowing for protocols to have a small but nonzero error probability $\epsilon > 0$. Indeed, looking at early examples of tasks, the advantage disappears for the distributed Deutsch-Jozsa problem~\cite{Buhrman1998}, while it remains for Raz's problem~\cite{Raz1999}. We prove in the Appendix that the one-way quantum communication complexity with bounded error for $EE_n$ scales as $\Omega(2^n)$, and thus that the exponential separation in communication complexity due to superposition of causal ordering persists when allowing for a nonzero error probability.

To show that it is possible to operationally distinguish quantum control of causal order from two-way communication one could introduce counters at the output ports of Alice's and Bob's laboratories, whose role is to count the number of uses of the channels. Such an argument has already been made in Ref.~\cite{Araujo2014_PRL} to justify a computational advantage. We can model a counter as a qutrit initially in the state $\ket{0}$, whose evolution when a system exits the laboratory is $\ket{i} \to \ket{i + 1 \, \mathrm{mod} \, 3}$, where $i \in \{0,1,2\}$. Then, for both one-way communication and the quantum switch, the counters of Alice and Bob will be in the state $\ket{1}$ at the end of the protocol; for genuine two-way communication, at least one of these counters will be in the final state $\ket{2}$. Therefore, the expectation value of the observables $N = \sum_{i=0}^2 \ket{i} \bra{i}$ for the counters allows us to distinguish realizations of the quantum switch, such as~\cite{Procopio2015}, from two-way quantum communication.

In conclusion, we have found a communication complexity task, the Exchange Evaluation game, for which a quantum superposition of the direction of communication --- the quantum switch --- results in an exponential saving in communication when compared to causally ordered quantum communication. An interesting feature of this game is that it is not a promise game, as are most known tasks for which quantum resources have an exponential advantage~\cite{Buhrman2010}.

In future work, it would be interesting to explore other information processing tasks for which the quantum switch -- or other causally indefinite processes -- may yield interesting advantages. For example, one could look at the uses of the quantum switch for secure distributed computation \cite{Yao1982, Lo1997, Buhrman2012, Liu2013}. Indeed, imagine that Alice and Bob both want to learn about the value of $EE_n$, in such a way that the other party does not learn about their inputs. They could achieve this goal by enlisting a third party and using the quantum switch with the $EE_n$ protocol.

We thank Ashley Montanaro for pointing out Ref.~\cite{Klauck2000} to us, used to establish the bounded-error advantage. We acknowledge support from the European Commission project RAQUEL (No. 323970); the Austrian Science Fund (FWF) through the Special Research Programme FoQuS, the Doctoral Programme CoQuS and Individual Project (No. 2462), and the John Templeton Foundation. P. A. G. is also supported by FQRNT (Quebec).


\bibliographystyle{abbrv}

\providecommand{\href}[2]{#2}\begingroup\raggedright\endgroup

\appendix

\section{VC-dimension bounds on the bounded error one-way quantum communication complexity}

In this section we show that if the protocol allows for some error probability, bounded by $\epsilon$ for all inputs, the one-way communication complexity of $EE_n$ still scales as $\Omega(2^n)$. As in Fig.~\ref{fig:causally_ordered}, we assume that Alice and Bob share unlimited prior entanglement, and that Alice sends a quantum state to Bob. We note that under the promise that Bob's input function is the zero function $g = 0$, the Exchange Evaluation game reduces to a random access code~\cite{Ambainis1999}, for which optimal bounds on the bounded error communication complexity are known~\cite{Nayak1998}. However, it is more straightforward to apply a bound that uses the concept of VC-dimension~\cite{Vapnik1971}.

\begin{definition}
\textit{VC-dimension.} Let $f:  X \times Y \to \{0, 1\}$. A subset $S \subseteq Y$ is shattered, if $\forall R \subseteq S, \exists x \in X$ such that
\begin{equation}
f(x,y) = \begin{cases}
1, & \text{if $\mathbf{y} \in R$}.\\
0, &  \text{if $\mathbf{y} \in S \setminus R$}.
\end{cases}
\end{equation}
The VC-dimension $VC(f)$ is the size of the largest shattered subset of $Y$.
\end{definition}

Given a function $f(x,y)$, we denote by $Q_\epsilon^{1}(f)$ the one-way (from Alice to Bob) bounded error quantum communication, where $\epsilon$ is the allowed worst-case error, and arbitrary prior shared entanglement is available. We make use of a theorem by Klauck (Theorem 3 of~\cite{Klauck2000}) that relates the bounded error quantum communication complexity of a function to its VC-dimension. 

\begin{theorem}
For all functions $f: X \times Y \to \{0,1\}$, $Q_\epsilon^1 (f) \geq \frac{1}{2} (1 - H(\epsilon)) VC(f)$, where $H(\epsilon)$ is the binary entropy $H(\epsilon) = \epsilon \log(\epsilon) + (1 - \epsilon) \log (1-\epsilon)$
\label{th:klauck}
\end{theorem}

Let us bound the VC-dimension of $EE_n: X \times Y \to \{0,1\}$, where $X = Y = \mathbb{Z}_2^n\times F_n$, by showing that $S = \{(\mathbf{y},g) | g = 0, \mathbf{y} \neq \mathbf{0} \} \subset Y$ is shattered. This is clear, since for any $R \subseteq S$, there exists the indicator function
\begin{equation}
f_R(\mathbf{y}) = \begin{cases}
1, & \text{if $(\mathbf{y},0) \in R$}.\\
0, & \text{otherwise},
\end{cases}
\end{equation}
so that $S$ is shattered.

Therefore $VC(EE_n) \geq |S| = 2^{n - 1}$, and Theorem~\ref{th:klauck} implies that the one-way quantum communication complexity $Q_\epsilon^1 ( EE_n) \geq (1 - H(\epsilon)) 2^{n - 2}$. This establishes that the number of communicated qubits scales exponentially with $n$ even in the bounded error case, so that the exponential separation between the quantum switch and one-way quantum communication continues to hold.
\end{document}